\documentclass[
preprint,
amsmath,amssymb,
aps,
]{revtex4-2}

\usepackage{graphicx}
\usepackage{dcolumn}
\usepackage{bm}
\usepackage{hyperref}

\usepackage{subcaption}

\usepackage{amsthm}

\usepackage{dcolumn}
\usepackage{bm}
\usepackage{hyperref}
\usepackage{subcaption}
\usepackage{braket}
\usepackage{caption}
\usepackage{diagbox}
\usepackage{booktabs}
\usepackage{adjustbox}
\usepackage{multirow}
\usepackage[mathlines]{lineno}
\usepackage{ragged2e}
\usepackage{pifont}
\usepackage{makecell}
\usepackage{algorithm}
\usepackage{algpseudocode}
\usepackage{soul}
\usepackage{cleveref}
\usepackage{tikz}
\usepackage{array}
\usepackage{centernot}
\usepackage{cancel}

\newtheorem{observation}{Observation}
\newtheorem{remark}{Remark}

\newtheorem{corollary}{Corollary}

\newtheorem{proposition}{Proposition}

\newcommand{\Tr}{\operatorname{tr}}
\newcommand{\tr}{\operatorname{tr}}

\newcommand{\W}{\mathcal{W}}
\newcommand{\eps}{\varepsilon}

\definecolor{codegreen}{rgb}{0,0.6,0}
\definecolor{codegray}{rgb}{0.5,0.5,0.5}
\definecolor{codepurple}{rgb}{0.58,0,0.82}
\definecolor{backcolour}{rgb}{0.95,0.95,0.92}

\begin{document}


\title{One-parameter counterexamples to the refined Bessis-Moussa-Villani conjecture}


\author{Hyunho Cha}
\email{ovalavo@snu.ac.kr}
\author{Jungwoo Lee}
\email{junglee@snu.ac.kr}
\affiliation{NextQuantum Innovation Research Center, Seoul National University, Seoul 08826, Republic of Korea}


\begin{abstract}
Positivity of matrix trace exponentials is a basic structural principle behind finite-temperature quantum statistical mechanics.
The Bessis-Moussa-Villani conjecture, a central manifestation of this principle, was proved by Stahl after an influential reformulation by Lieb and Seiringer. A later refinement asks whether the normalized average over all words with $n$ letters $A$ and $m$ letters $B$ is always bounded above by $\mathrm{tr}(A^nB^m)$ and below by $\mathrm{tr}\exp(n\log A+m\log B)$. In this work, we study a specific one-parameter family $(A_x, B_x)$ and show that the correct small-$x$ invariant of a word is not its degree of fragmentation, but a weighted shortest-bridge cost on its cyclic run decomposition. Our results yield a class of counterexamples to the suggested refinement. Remarkably, the ratio of the normalized word average to the trace $\mathrm{tr}(A^nB^m)$ can become arbitrarily large.
\end{abstract}

\maketitle


\section{Introduction}

The Bessis-Moussa-Villani (BMV) conjecture was introduced in 1975 in the setting of quantum statistical mechanics~\cite{bessis1975monotonic}. One convenient form of the conjecture concerns the trace exponential
$
\lambda \mapsto \Tr e^{A-\lambda B},
$
with $A$ Hermitian and $B\succeq 0$ \cite{mehta1976integral, le1980representation, moussa2000representation, fannes2003perturbation, drmota2005hyper, hansen2006trace, heinavaara2025tracial}. Lieb and Seiringer showed that the BMV conjecture is equivalent to the statement that every coefficient of the polynomial
$
\Tr(A+tB)^p
$
is nonnegative whenever $A,B\succeq 0$ and $p\in\mathbb{N}$~\cite{lieb2004equivalent, hillar2005positivity, johnson2005principal, hagele2007proof, hillar2007advances, friedland2008remarks, klep2008sums, landweber2009d, cafuta2010note, collins2010sum, fleischhack2010asymptotic, burgdorf2011sums, lieb2012further}. Stahl later proved the conjecture~\cite{stahl2013proof, eremenko2015herbert}. The problem is important in quantum theory because trace exponentials are finite-dimensional partition functions of Gibbs states, and the BMV positivity statement implies inequalities for derivatives of thermodynamic partition functions \cite{bessis1975monotonic,iqoqi_problem40}. Thus BMV lies at the intersection of operator inequalities, noncommutative positivity, and finite-temperature quantum many-body theory \cite{lieb1973convex, petz1994survey}.

A natural refinement, attributed to Daniel H\"agele and listed as an Open Quantum Problem, asks for more than mere nonnegativity of the average coefficient \cite{adnane2017refinement, iqoqi_problem40}. For positive semidefinite matrices $A,B$, define $p_{n,m}(A,B)$ by
\begin{equation}
\binom{n+m}{n}\,p_{n,m}(A,B)
=
[t^ns^m] \, \Tr(tA+sB)^{n+m},
\label{eq:defpnm}
\end{equation}
where $[t^ns^m]$ denotes coefficient extraction. Equivalently,
\[
p_{n,m}(A,B):=\frac{1}{\binom{n+m}{n}}\sum_{W\in\W_{n,m}}\tr W(A,B),
\]
where \(\W_{n,m}\) is the set of all linear words in the letters \(A,B\) containing \(n\) copies of \(A\) and \(m\) copies of \(B\).
The refinement asks whether
\begin{equation}
\Tr(A^nB^m)\ge p_{n,m}(A,B)\ge \Tr\exp(n\log A+m\log B)
\label{eq:refinement}
\end{equation}
always holds. When $A$ and $B$ commute, equality is immediate. When $n=m=1$, the lower bound reduces to the Golden--Thompson inequality \cite{golden1965lower, thompson1965inequality, lenard1971generalization, araki1973golden, thompson2014golden}. The intuition behind Eq.~\eqref{eq:refinement} is that clustering equal letters should increase the trace, while increasingly fragmented products should push the trace downward \cite{johnson2002eigenvalues, hillar2005positive}. We show that even the averaged inequality on the left of Eq.~\eqref{eq:refinement} fails.

In this work, we consider a positive semidefinite family $(A_x,B_x)$ and isolate the structural mechanism behind this counterexample. The key point is that the family admits a rank-one projection normal form. Once this is written down, the small-\(x\) asymptotics of every word become a combinatorial optimization problem on its run decomposition. That optimization problem has no monotone dependence on the number of runs, and this is precisely why the clustering (or fragmentation) intuition breaks down. Our findings not only confirm that the inequality is subject to violation, but also reveal that the ratio $p_{n,m}(A_x,B_x)/\Tr(A_x^n B_x^m)$ can grow without bound.

The main conceptual conclusions are the following.
\begin{enumerate}
\item The counterexample sits near a \emph{commuting} limit, so the blow-up is not caused by large noncommutativity.
\item Every coefficient that appears in a word trace is nonnegative, so the mechanism is not cancellation.
\item The correct invariant is a weighted shortest-bridge cost, not the crude number of alternations.
\item The averaged trace is dominated by a very small family of exceptional words.
\end{enumerate}

\section{Projection normal form}

For $x\ge 0$, set
$$
A_x=\begin{pmatrix}
1&0&0\\
0&x&-x\\
0&-x&x
\end{pmatrix},
\qquad
B_x=\begin{pmatrix}
x&-x&0\\
-x&x&0\\
0&0&1
\end{pmatrix}.
$$
Set \(\eps:=2x\), and define the unit vectors
\begin{gather*}
p=e_1,
\qquad
u=\frac{(0,-1,1)}{\sqrt2},
\qquad
v=\frac{(1,-1,0)}{\sqrt2},
\qquad
q=e_3.
\end{gather*}
Let
$
P=|p\rangle\langle p|,
$
$
U=|u\rangle\langle u|,
$
$
V=|v\rangle\langle v|,
$
and
$
Q=|q\rangle\langle q|.
$
Explicitly,
\begin{gather*}
P=
\begin{pmatrix}
1&0&0\\0&0&0\\0&0&0
\end{pmatrix},
\qquad
U=\frac12
\begin{pmatrix}
0&0&0\\0&1&-1\\0&-1&1
\end{pmatrix},
\qquad
V=\frac12
\begin{pmatrix}
1&-1&0\\-1&1&0\\0&0&0
\end{pmatrix},
\qquad
Q=
\begin{pmatrix}
0&0&0\\0&0&0\\0&0&1
\end{pmatrix}.
\end{gather*}

\begin{observation}\label{obs:normal-form}
The matrices \(A_x,B_x\) admit the decomposition
$
A_x=P+\eps U
$
and
$
B_x=\eps V+Q.
$
Moreover, \(P,U\) are orthogonal rank-one projections, \(V,Q\) are orthogonal rank-one projections, and the only nonzero overlaps among \(P,U,V,Q\) are
\[
\langle p,v\rangle=\frac{1}{\sqrt2},
\qquad
\langle v,u\rangle=\frac12,
\qquad
\langle u,q\rangle=\frac{1}{\sqrt2},
\]
and their reverses. Equivalently, the compatibility graph is the path
\[
P\longleftrightarrow V\longleftrightarrow U\longleftrightarrow Q,
\]
and all other pairings vanish.
\end{observation}

The clustered product \(A_x^nB_x^m\) is now completely transparent.

\begin{proposition}\label{prop:cluster-exact}
For every \(n,m\ge 1\),
\[
\tr(A_x^nB_x^m)=2^{m-1}x^m+2^{n-1}x^n+2^{n+m-2}x^{n+m}.
\]
In particular, for \(n=m=5\),
\[
\tr(A_x^5B_x^5)=32x^5+256x^{10}.
\]
\end{proposition}

\begin{proof}
Because \(P,U\) are orthogonal projections,
\[
A_x^n=(P+\eps U)^n=P+\eps^nU.
\]
Likewise,
\[
B_x^m=(\eps V+Q)^m=\eps^mV+Q.
\]
Therefore
\begin{align*}
A_x^nB_x^m & = (P+\eps^nU)(\eps^mV+Q)\\
& =\eps^mPV+\eps^nUQ+\eps^{n+m}UV,
\end{align*}
since \(PQ=0\). Taking traces and using
\begin{gather*}
\tr(PV)=|\langle p,v\rangle|^2=\frac12,
\qquad
\tr(UQ)=|\langle u,q\rangle|^2=\frac12,
\qquad
\tr(UV)=|\langle u,v\rangle|^2=\frac14,
\end{gather*}
we obtain
\[
\tr(A_x^nB_x^m)=\frac12\eps^m+\frac12\eps^n+\frac14\eps^{n+m}.
\]
Substituting \(\eps=2x\) gives the result.
\end{proof}

\begin{remark}\label{rem:boundary}
At \(x=0\) one has \(A_0=P\) and \(B_0=Q\). These commute, but \(PQ=0\). Hence every mixed word has trace zero at the limiting commuting pair. The small-\(x\) problem is therefore a competition of vanishing orders near a singular boundary point of the positive semidefinite cone.
\end{remark}

\section{Admissible projection walks and the leading exponent of a word}

Let \(W\) be a word containing both letters. Because trace is invariant under cyclic rotation, we may rotate \(W\) so that it starts with an \(A\)-run and write it in cyclic run form as
\[
W=A^{a_1}B^{b_1}A^{a_2}B^{b_2}\cdots A^{a_r}B^{b_r},
\qquad a_i,b_i\ge 1.
\]
Here \(r\) is the number of \(A\)-runs (equivalently, the number of \(B\)-runs).

\begin{observation}\label{obs:rank-one-trace}
If \(R_{\xi}=|\xi\rangle\langle\xi|\) is the rank-one projection onto a unit vector \(\xi\), then for any unit vectors \(\xi_1,\dots,\xi_k\),
\[
\tr\bigl(R_{\xi_1}R_{\xi_2}\cdots R_{\xi_k}\bigr)
=\prod_{j=1}^k\langle\xi_j,\xi_{j+1}\rangle,
\qquad \xi_{k+1}:=\xi_1.
\]
In particular, the trace is nonzero if and only if every adjacent overlap \(\langle\xi_j,\xi_{j+1}\rangle\) is nonzero.
\end{observation}

Using the projection decomposition run by run gives
\[
A_x^{a_i}=P+\eps^{a_i}U,
\qquad
B_x^{b_i}=\eps^{b_i}V+Q.
\]
Hence every word expands as a sum over choices of one projection from each run.

\begin{proposition}\label{prop:admissible-expansion}
Let
\[
W=A^{a_1}B^{b_1}\cdots A^{a_r}B^{b_r}.
\]
Then
\[
\tr W(A_x,B_x)
=\sum_{\substack{\sigma_i\in\{P,U\}\\ \tau_i\in\{V,Q\}}}
\eps^{\mathrm{wt}(\sigma,\tau)}
\tr(\sigma_1\tau_1\sigma_2\tau_2\cdots\sigma_r\tau_r),
\]
where
\[
\mathrm{wt}(\sigma,\tau):=\sum_{\sigma_i=U}a_i+\sum_{\tau_i=V}b_i.
\]
Every nonzero summand is strictly positive. Consequently, \(\tr W(A_x,B_x)\) is a polynomial in \(x\) with nonnegative coefficients.
\end{proposition}

\begin{proof}
The expansion is obtained by substituting \(A_x^{a_i}=P+\eps^{a_i}U\) and \(B_x^{b_i}=\eps^{b_i}V+Q\) and multiplying out. If a summand is nonzero, then by Observation~\ref{obs:rank-one-trace} its trace is the product of adjacent overlaps of the chosen vectors. By Observation~\ref{obs:normal-form}, every nonzero overlap is positive, namely \(1/\sqrt2\) or \(1/2\). Therefore every nonzero summand is strictly positive.
\end{proof}

This motivates the following definition.
\[
\kappa(W):=\min\bigl\{d\ge 0:\ [x^d]\tr W(A_x,B_x)\neq 0\bigr\}.
\]
Since \(\eps=2x\), the same exponent is obtained from the \(\eps\)-expansion.

\begin{proposition}[weighted shortest-bridge formula]\label{prop:kappa-formula}
Let
\[
W=A^{a_1}B^{b_1}\cdots A^{a_r}B^{b_r}.
\]
For a subset \(S\subseteq\{1,\dots,r\}\), define
\[
\Gamma(S):=\{i\in\{1,\dots,r\}: i\in S\ \text{or}\ i+1\in S\},
\]
where indices are taken modulo \(r\). Then
\[
\kappa(W)=\min_{S\subseteq\{1,\dots,r\}}
\left(
\sum_{i\notin S}a_i+\sum_{i\in\Gamma(S)}b_i
\right).
\]
\end{proposition}

\begin{proof}
Fix an admissible choice \((\sigma,\tau)\). Let \(S=\{i: \sigma_i=P\}\). Then the \(A\)-runs outside \(S\) are exactly those assigned to \(U\), so they contribute \(\sum_{i\notin S}a_i\) to the weight.

Now consider a fixed \(B\)-run indexed by \(i\). In the cyclic product
$
\sigma_1\tau_1\sigma_2\tau_2\cdots\sigma_r\tau_r,
$
this \(B\)-run sits between \(\sigma_i\) and \(\sigma_{i+1}\). If either \(i\in S\) or \(i+1\in S\), then one of these neighboring \(A\)-runs equals \(P\). Since \(P\) is compatible only with \(V\), admissibility forces \(\tau_i=V\). Thus every admissible assignment with \(P\)-set \(S\) has weight at least
$
\sum_{i\notin S}a_i+\sum_{i\in\Gamma(S)}b_i.
$

Conversely, for any fixed \(S\), define
\[
\sigma_i=
\begin{cases}
P,& i\in S,\\
U,& i\notin S,
\end{cases}
\qquad
\tau_i=
\begin{cases}
V,& i\in\Gamma(S),\\
Q,& i\notin\Gamma(S).
\end{cases}
\]
If \(i\in\Gamma(S)\), then \(\tau_i=V\), and \(V\) is compatible with both \(P\) and \(U\). If \(i\notin\Gamma(S)\), then neither \(i\) nor \(i+1\) lies in \(S\), so \(\sigma_i=\sigma_{i+1}=U\), and \(Q\) is compatible with \(U\). Hence this assignment is admissible and has weight
$
\sum_{i\notin S}a_i+\sum_{i\in\Gamma(S)}b_i.
$
Taking the minimum over \(S\) proves the formula.
\end{proof}

\begin{remark}\label{rem:interpretation-kappa}
Proposition~\ref{prop:kappa-formula} is the central structural statement. The subset \(S\) specifies which \(A\)-runs stay in the ``macroscopic'' sector \(P\). All remaining \(A\)-runs must go through the ``small'' sector \(U\), and each selected \(P\)-run forces its two neighboring \(B\)-runs through the ``small'' sector \(V\). Thus \(\kappa(W)\) is a weighted closed-neighborhood cost on the run cycle, not a monotone function of the number of runs.
\end{remark}

\section{Failure of monotone fragmentation}

We now compare three words with the same total content \(n=m=5\): the clustered word \(A^5B^5\), the bridge word \(A^3BAB^3AB\), and the fully alternating word \((AB)^5\).

\begin{proposition}\label{prop:three-examples}
For the family \((A_x,B_x)\),
\[
\kappa(A^5B^5)=5,
\qquad
\kappa(A^3BAB^3AB)=4,
\qquad
\kappa((AB)^5)=5.
\]
Moreover,
\begin{align*}
\tr(A_x^5B_x^5)&=32x^5+256x^{10},\\
\tr(A_x^3B_xA_xB_x^3A_xB_x)&=x^4+O(x^5),\\
\tr((A_xB_x)^5)&=2x^5+O(x^6).
\end{align*}
\end{proposition}

\begin{proof}
For \(A^5B^5\), Proposition~\ref{prop:cluster-exact} gives the formula, hence \(\kappa(A^5B^5)=5\).

For \(A^3BAB^3AB\), the run lengths are
\[
(a_1,a_2,a_3)=(3,1,1),
\qquad
(b_1,b_2,b_3)=(1,3,1).
\]
Take \(S=\{1\}\). Then
$
\Gamma(S)=\{1,3\},
$
so Proposition~\ref{prop:kappa-formula} gives
\[
\kappa(A^3BAB^3AB)\le (1+1)+(1+1)=4.
\]
On the other hand, any admissible assignment that uses both \(P\) and \(Q\) must contain at least two occurrences of \(U\) and two occurrences of \(V\), because the compatibility graph is the path \(P-V-U-Q\) and a closed walk visiting both endpoints has to cross the middle edge in both directions. Hence the cost is at least \(4\). Therefore \(\kappa(A^3BAB^3AB)=4\).

The leading assignment is uniquely obtained from \(S=\{1\}\): the long \(A\)-run is assigned to \(P\), the long \(B\)-run to \(Q\), and the four singleton runs are assigned to the bridge sectors \(U,V,U,V\). Its leading coefficient is
\[
\eps^4\tr(PVUQUV)
=\eps^4\,\langle p,v\rangle\langle v,u\rangle\langle u,q\rangle\langle q,u\rangle\langle u,v\rangle\langle v,p\rangle
=\frac{\eps^4}{16}=x^4.
\]
Hence \(\tr(A_x^3B_xA_xB_x^3A_xB_x)=x^4+O(x^5)\).

For \((AB)^5\), every run has length \(1\). Proposition~\ref{prop:kappa-formula} becomes
\[
\kappa((AB)^5)=\min_{S\subseteq\{1,\dots,5\}}\bigl(5-|S|+|\Gamma(S)|\bigr).
\]
If \(S\) is a nonempty proper subset of the cyclic group \(\mathbb Z/5\mathbb Z\), then \(\Gamma(S)=S\cup(S-1)\) strictly contains \(S\), implying that \(|\Gamma(S)|>|S|\). Thus the minimum \(5\) occurs only for \(S=\varnothing\) and \(S=\{1,2,3,4,5\}\). These two minimizers correspond to the assignments
$$
(U,Q,U,Q,U,Q,U,Q,U,Q)
$$
and
$$
(P,V,P,V,P,V,P,V,P,V).
$$
Each contributes
\begin{align*}
\eps^5\tr((UQ)^5)=\eps^5\Bigl(\frac12\Bigr)^5=x^5,\\
\eps^5\tr((PV)^5)=\eps^5\Bigl(\frac12\Bigr)^5=x^5.
\end{align*}
Hence \(\tr((A_xB_x)^5)=2x^5+O(x^6)\).
\end{proof}

\begin{table}[t]
\centering
\caption{\justifying The bridge word has intermediate fragmentation but a strictly smaller leading exponent. Thus there is no monotone ordering from ``more clustered'' to ``more fragmented.''}
\label{tab:comparison}
\begin{tabular*}{\textwidth}{@{\extracolsep{\fill}}lccc@{}}
\toprule
Word \(W\) & \makecell[c]{Number of run pairs $r$} & \(\kappa(W)\) & \makecell[c]{Leading term of \(\tr W(A_x,B_x)\)} \\
\midrule
\(A^5B^5\) & 1 & 5 & \(32x^5\) \\
\(A^3BAB^3AB\) & 3 & 4 & \(x^4\) \\
\((AB)^5\) & 5 & 5 & \(2x^5\) \\
\bottomrule
\end{tabular*}
\end{table}

Table~\ref{tab:comparison} summarizes the outcome. The table falsifies the heuristic in the strongest possible way. The fully alternating word is indeed smaller than the clustered word, but only by a coefficient. The decisive word is neither the most clustered nor the most fragmented. It is the one that localizes the ``small'' sectors on four singleton bridge letters.

\section{Classification of the order-\texorpdfstring{$x^4$}{x4} words}

We now classify which words achieve leading order \(x^4\).

\begin{proposition}\label{prop:classification}
Assume \(n,m\ge 5\), and let \(W\) be any word with \(n\) letters \(A\) and \(m\) letters \(B\). Then \(\kappa(W)\ge 4\). Moreover, \(\kappa(W)=4\) if and only if, up to cyclic rotation,
$
W=A^{n-2}BAB^{m-2}AB.
$
\end{proposition}

\begin{proof}
Let \((\sigma,\tau)\) be an admissible assignment achieving the minimum defining \(\kappa(W)\).

If no \(P\) occurs, then every \(A\)-run is assigned to \(U\), so the total cost is at least \(n\ge 5\). If no \(Q\) occurs, then every \(B\)-run is assigned to \(V\), so the total cost is at least \(m\ge 5\). Therefore any assignment with cost at most \(4\) must use both \(P\) and \(Q\).

Now consider the cyclic sequence of chosen projections. Because the compatibility graph is the path
\[
P-V-U-Q,
\]
a closed admissible walk that visits both endpoints must contain at least two occurrences of \(U\) and two occurrences of \(V\): one needs a segment \(P\to V\to U\to Q\) to reach \(Q\), and another segment \(Q\to U\to V\to P\) to return. Since every run length is at least \(1\), the cost is therefore at least
$
1+1+1+1=4.
$
Hence \(\kappa(W)\ge 4\).

Suppose now that \(\kappa(W)=4\). Then the previous paragraph shows that the minimizing assignment contains exactly two \(U\)-runs and exactly two \(V\)-runs, and each of those four runs must have length \(1\). Therefore, the cyclic sequence of chosen projections has the unique form
\[
P\cdots P\,V\,U\,Q\cdots Q\,U\,V\,P\cdots P.
\]
Thus there is one contiguous \(P\)-block, one contiguous \(Q\)-block, two singleton \(U\)-blocks, and two singleton \(V\)-blocks. Translating back to letters, this means that, up to cyclic rotation, the word is exactly
\[
A^{n-2}BAB^{m-2}AB.
\]
Conversely, this word admits the assignment described in the proof of Proposition~\ref{prop:three-examples}, so its leading exponent is \(4\).
\end{proof}

\begin{figure}
    \centering
    \includegraphics[width=0.5\linewidth]{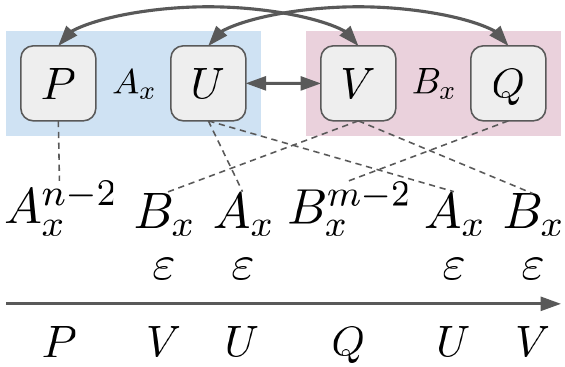}
    \caption{Illustration of the leading order \(x^4\) for the bridge pattern $A_x^{n-2}B_xA_xB_x^{m-2}A_xB_x$.}
    \label{fig:bridge_pattern}
\end{figure}

A visualization of the bridge pattern is provided in Figure~\ref{fig:bridge_pattern}. The preceding proposition immediately controls the averaged trace.

\begin{corollary}\label{cor:average}
For every \(n,m\ge 5\),
\[
p_{n,m}(A_x,B_x)=\frac{n+m}{\binom{n+m}{n}}x^4+O(x^5).
\]
In particular, for \((n,m)=(5,5)\),
\[
p_{5,5}(A_x,B_x)=\frac{5}{126}x^4+O(x^5).
\]
\end{corollary}

\begin{proof}
By Proposition~\ref{prop:classification}, the only words with leading order \(x^4\) are the cyclic shifts of
\[
W_{n,m}:=A^{n-2}BAB^{m-2}AB.
\]
Because \(n,m\ge 5\), this word has a unique long \(A\)-block of length \(n-2\ge 3\) and a unique long \(B\)-block of length \(m-2\ge 3\). All other runs are singletons. Hence no nontrivial cyclic shift fixes the word, so the \(n+m\) cyclic shifts are all distinct as linear words.

Each such word has the same leading coefficient, namely
\[
\eps^4\tr(PVUQUV)=\frac{\eps^4}{16}=x^4.
\]
Every other word is \(O(x^5)\). Summing over all \(\binom{n+m}{n}\) linear words therefore yields
\[
\sum_{W\in\W_{n,m}}\tr W(A_x,B_x)=(n+m)x^4+O(x^5).
\]
Dividing by \(\binom{n+m}{n}\) proves the formula.
\end{proof}

\begin{remark}\label{rem:rare-word}
For \((n,m)=(5,5)\), there are \(\binom{10}{5}=252\) words in total, and \(10\) of them contribute at order \(x^4\). The remaining \(242\) words are all \(O(x^5)\). Thus the averaged trace is controlled by a tiny exceptional family of ``bridge'' words. Since Proposition~\ref{prop:admissible-expansion} shows that all coefficients are nonnegative, this dominance cannot be cancelled away.
\end{remark}

Combining Corollary~\ref{cor:average} with Proposition~\ref{prop:cluster-exact} gives the asymptotic ratio.

\begin{corollary}\label{cor:ratio}
Let \(\ell:=\min\{n,m\}\), and define
\[
d_{n,m}:=
\begin{cases}
2^{\ell-1},& n\ne m,\\
2^{\ell},& n=m.
\end{cases}
\]
Then for every \(n,m\ge 5\),
\[
\frac{p_{n,m}(A_x,B_x)}{\tr(A_x^nB_x^m)}
=\frac{n+m}{d_{n,m}\binom{n+m}{n}}x^{4-\ell}+O\bigl(x^{5-\ell}\bigr).
\]
Hence the ratio diverges as \(x\to 0^+\) whenever \(\min\{n,m\}>4\).
\end{corollary}

\begin{proof}
By Proposition~\ref{prop:cluster-exact}, the denominator has leading term \(d_{n,m}x^{\ell}\). By Corollary~\ref{cor:average}, the numerator has leading term \(\frac{n+m}{\binom{n+m}{n}}x^4\). Dividing the two asymptotic expansions gives the result.
\end{proof}

For \(n=m=5\) this becomes
\[
\frac{p_{5,5}(A_x,B_x)}{\tr(A_x^5B_x^5)}
=\frac{5}{4032}\frac{1}{x}+O(1).
\]

\section{Failure of the original intuition}

The preceding analysis suggests a more precise replacement for the clustering (or fragmentation) heuristic.

\paragraph{The right variable is vanishing order, not visual fragmentation.}
Near \(x=0\), every mixed trace tends to zero because \(A_0=P\), \(B_0=Q\), and \(PQ=0\). Therefore different words are compared first by their leading exponent. The exponent is the weighted shortest-bridge cost \(\kappa(W)\), not the number of runs. The three examples from Table~\ref{tab:comparison} already show that the map
\[
\text{``more fragmented''}\longmapsto \text{``smaller trace''}
\]
is not monotonic.

\paragraph{The failure occurs near a commuting pair.}
A natural stability intuition would say that if \(A_x\) and \(B_x\) almost commute, then any ordering inequality between different words should fail only mildly, if at all. But here
\[
[A_x,B_x]=
\begin{pmatrix}
0&x(x-1)&-x^2\\
x(1-x)&0&x(x-1)\\
x^2&x(1-x)&0
\end{pmatrix},
\]
so
\begin{align*}
\|[A_x,B_x]\|_\mathrm{F}^2 & =4x^2(1-x)^2+2x^4=4x^2+O(x^3),\\
\|[A_x,B_x]\|_\mathrm{F} & =2x+O(x^2).
\end{align*}
Thus the pair becomes commutative with a linear rate as \(x\to0\), while Corollary~\ref{cor:ratio} shows that the ratio of the averaged trace to the clustered trace diverges with a negative power of \(x\). The commuting limit is therefore singular for this comparison problem.

\paragraph{Averaging is governed by rare exceptional words.}
The average \(p_{n,m}(A_x,B_x)\) is not controlled by the bulk of the word ensemble. Instead it is controlled by the exceptional words that achieve the minimal bridge cost. For \((n,m)=(5,5)\), ten bridge words dominate an average over 252 words. This is the opposite of the intuition that the average should resemble a ``typical'' fragmented word.

\paragraph{The right replacement principle is a shortest-bridge principle.}
The projection normal form shows that the only way to connect the macroscopic \(A\)-sector \(P\) to the macroscopic \(B\)-sector \(Q\) is through the two bridge sectors \(V\) and \(U\). A word becomes large when it organizes its run lengths so that these bridges are paid for on very short runs. The extremal architecture is therefore not maximal clustering and not maximal alternation, but the mixed bridge pattern
$
A^{n-2}BAB^{m-2}AB.
$
That pattern minimizes the bridge cost by concentrating the expensive transitions on four singleton letters.

\section{Discussion}

Our results show that the failure of the first inequality in Eq.~\eqref{eq:refinement} is not an isolated numerical accident but part of a simple one-parameter mechanism. Conceptually, this is also physically significant because the original BMV conjecture has direct implications for quantum statistical mechanics. For a finite-dimensional Hamiltonian $H$, the partition function is $Z(\beta)=\Tr(e^{-\beta H})$, and BMV concerns positivity properties of such trace exponentials under positive perturbations~\cite{bessis1975monotonic,lieb2004equivalent,stahl2013proof}. The proposed refinement aimed to strengthen positivity into an ordering among different noncommutative products, thereby connecting Gibbs-type expressions, Golden--Thompson-type phenomena, and averaged words in positive semidefinite matrices \cite{cohen1982eigenvalue, cohen1988spectral, so1992equality, hiai1993golden, ando1994log, friedland1994product, sutter2017multivariate, randrianantoanina2024triple}. The one-parameter counterexamples derived here show that this stronger ordering principle fails already in dimension $3$ and already at $n=m=5$.

\section*{Acknowledgments}

This work is in part supported by the National Research Foundation of Korea (NRF, RS-2024-00451435 (20\%), RS-2024-00413957 (20\%)), Institute of Information \& communications Technology Planning \& Evaluation (IITP, RS-2025-02305453 (15\%), RS-2025-02273157 (15\%), RS-2025-25442149 (15\%), RS-2021-II211343 (15\%)) grant funded by the Ministry of Science and ICT (MSIT), Institute of New Media and Communications (INMAC), and the BK21 FOUR program of the Education, Artificial Intelligence Graduate School Program (Seoul National University), and Research Program for Future ICT Pioneers, Seoul National University in 2026.

\bibliography{apssampv1}

@article{bessis1975monotonic,
  title={Monotonic converging variational approximations to the functional integrals in quantum statistical mechanics},
  author={Bessis, Daniel and Moussa, Pierre and Villani, Matteo},
  journal={Journal of Mathematical Physics},
  volume={16},
  number={11},
  pages={2318--2325},
  year={1975},
  publisher={American Institute of Physics}
}

@article{lieb2004equivalent,
  title={Equivalent forms of the {B}essis--{M}oussa--{V}illani conjecture},
  author={Lieb, Elliott H and Seiringer, Robert},
  journal={Journal of Statistical Physics},
  volume={115},
  number={1},
  pages={185--190},
  year={2004},
  publisher={Springer}
}

@article{hagele2007proof,
  title={{Proof of the cases $p\le7$ of the Lieb-Seiringer formulation of the Bessis-Moussa-Villani conjecture}},
  author={H{\"a}gele, Daniel},
  journal={Journal of Statistical Physics},
  volume={127},
  number={6},
  pages={1167--1171},
  year={2007},
  publisher={Springer}
}

@misc{iqoqi_problem40,
  author       = {{IQOQI Vienna}},
  title        = {{Problem 40: Refinement of the Bessis-Moussa-Villani conjecture}},
  howpublished = {List of Open Quantum Problems},
  url          = {https://oqp.iqoqi.oeaw.ac.at/refinement-of-the-bessis-moussa-villani-conjecture},
  note         = {{Accessed: 2026-03-20}}
}

@ARTICLE{stahl2013proof,
  title     = "Proof of the {BMV} conjecture",
  author    = "Stahl, Herbert R",
  journal   = "Acta Math.",
  volume    =  211,
  number    =  2,
  pages     = "255--290",
  year      =  2013
}

@misc{adnane2017refinement,
  author       = {Mostafa Adnane},
  title        = {{Refinement of the Bessis-Moussa-Villani Conjecture}},
  year         = {2017}
}

@article{hillar2007advances,
  title={{Advances on the Bessis--Moussa--Villani trace conjecture}},
  author={Hillar, Christopher J},
  journal={Linear Algebra and its Applications},
  volume={426},
  number={1},
  pages={130--142},
  year={2007},
  publisher={Elsevier}
}

@article{hillar2005positivity,
  title={On the positivity of the coefficients of a certain polynomial defined by two positive definite matrices},
  author={Hillar, Christopher J and Johnson, Charles R},
  journal={Journal of Statistical Physics},
  volume={118},
  number={3},
  pages={781--789},
  year={2005},
  publisher={Springer}
}

@article{burgdorf2011sums,
  title={{Sums of Hermitian squares as an approach to the BMV conjecture}},
  author={Burgdorf, Sabine},
  journal={Linear and Multilinear Algebra},
  volume={59},
  number={1},
  pages={1--9},
  year={2011},
  publisher={Taylor \& Francis}
}

@article{klep2008sums,
  title={{Sums of Hermitian squares and the BMV conjecture}},
  author={Klep, Igor and Schweighofer, Markus},
  journal={Journal of Statistical Physics},
  volume={133},
  number={4},
  pages={739--760},
  year={2008},
  publisher={Springer}
}

@article{fleischhack2010asymptotic,
  title={{Asymptotic positivity of Hurwitz product traces: Two proofs}},
  author={Fleischhack, Christian and Friedland, Shmuel},
  journal={Linear Algebra and its Applications},
  volume={432},
  number={6},
  pages={1363--1383},
  year={2010},
  publisher={Elsevier}
}

@article{lieb2012further,
  title={{Further implications of the Bessis--Moussa--Villani conjecture}},
  author={Lieb, Elliott H and Seiringer, Robert},
  journal={Journal of Statistical Physics},
  volume={149},
  number={1},
  pages={86--91},
  year={2012},
  publisher={Springer}
}

@article{landweber2009d,
  title={{On D. H{\"a}gele’s approach to the Bessis--Moussa--Villani conjecture}},
  author={Landweber, Peter S and Speer, Eugene R},
  journal={Linear Algebra and its Applications},
  volume={431},
  number={8},
  pages={1317--1324},
  year={2009},
  publisher={Elsevier}
}

@article{collins2010sum,
  title={{Sum-of-squares results for polynomials related to the Bessis--Moussa--Villani conjecture}},
  author={Collins, Beno{\^\i}t and Dykema, Kenneth J and Torres-Ayala, Francisco},
  journal={Journal of Statistical Physics},
  volume={139},
  number={5},
  pages={779--799},
  year={2010},
  publisher={Springer}
}

@article{johnson2005principal,
  title={{Principal minor sums of $(A+ tB)^m$}},
  author={Johnson, Charles R and Leichenauer, Stefan and McNamara, Peter and Costas, Roberto},
  journal={Linear Algebra and its Applications},
  volume={411},
  pages={386--389},
  year={2005},
  publisher={Elsevier}
}

@article{cafuta2010note,
  title={{A note on the nonexistence of sum of squares certificates for the Bessis--Moussa--Villani conjecture}},
  author={Cafuta, Kristijan and Klep, Igor and Povh, Janez},
  journal={Journal of Mathematical Physics},
  volume={51},
  number={8},
  year={2010},
  publisher={AIP Publishing}
}

@article{friedland2008remarks,
  title={{Remarks on BMV conjecture}},
  author={Friedland, Shmuel},
  journal={arXiv preprint arXiv:0804.3948},
  year={2008}
}

@article{hansen2006trace,
  title={{Trace functions as Laplace transforms}},
  author={Hansen, Frank},
  journal={Journal of Mathematical Physics},
  volume={47},
  number={4},
  year={2006},
  publisher={AIP Publishing}
}

@article{moussa2000representation,
  title={{On the representation of $\mathrm{Tr}(e^{(A-\lambda B)})$ as a Laplace transform}},
  author={Moussa, Pierre},
  journal={Reviews in Mathematical Physics},
  volume={12},
  number={04},
  pages={621--655},
  year={2000},
  publisher={World Scientific}
}

@article{le1980representation,
  title={{Representation of the function $\mathrm{Tr}(\exp(A-\lambda B))$ as a Laplace transform with positive weight and some matrix inequalities}},
  author={Le Couteur, KJ},
  journal={Journal of Physics A: Mathematical and General},
  volume={13},
  number={10},
  pages={3147--3159},
  year={1980}
}

@article{mehta1976integral,
  title={{On an integral representation of the function $\mathrm{Tr}(\exp(A-\lambda B))$}},
  author={Mehta, ML and Kumar, Kailash},
  journal={Journal of Physics A: Mathematical and General},
  volume={9},
  number={2},
  pages={197--206},
  year={1976}
}

@article{fannes2003perturbation,
  title={{Perturbation of Wigner matrices and a conjecture}},
  author={Fannes, Mark and Petz, D{\'e}nes},
  journal={Proceedings of the American Mathematical Society},
  volume={131},
  number={7},
  pages={1981--1988},
  year={2003}
}

@article{drmota2005hyper,
  title={{A hyper-geometric approach to the BMV-conjecture}},
  author={Drmota, Michael and Schachermayer, Walter and Teichmann, Josef},
  journal={Monatshefte f{\"u}r Mathematik},
  volume={146},
  number={3},
  pages={179--201},
  year={2005},
  publisher={Springer}
}

@article{heinavaara2025tracial,
  title={Tracial joint spectral measures},
  author={Hein{\"a}vaara, Otte},
  journal={Inventiones Mathematicae},
  volume={239},
  number={2},
  pages={505--526},
  year={2025},
  publisher={Springer}
}

@article{eremenko2015herbert,
  title={{Herbert Stahl's proof of the BMV conjecture}},
  author={Eremenko, Alexandre},
  journal={Sbornik: Mathematics},
  volume={206},
  number={1},
  pages={87--92},
  year={2015},
  publisher={London Mathematical Society, Turpion Ltd and the Russian Academy of Sciences}
}

@article{golden1965lower,
  title={{Lower bounds for the Helmholtz function}},
  author={Golden, Sidney},
  journal={Physical Review},
  volume={137},
  number={4B},
  pages={B1127},
  year={1965},
  publisher={APS}
}

@article{thompson1965inequality,
  title={Inequality with applications in statistical mechanics},
  author={Thompson, Colin J},
  journal={Journal of Mathematical Physics},
  volume={6},
  number={11},
  pages={1812--1813},
  year={1965},
  publisher={American Institute of Physics}
}

@article{lenard1971generalization,
  title={{Generalization of the Golden--Thompson Inequality $\mathrm{Tr}(e^A e^B) \geqq \mathrm{Tr} e^{A+B}$}},
  author={Lenard, A and Halmos, PR},
  journal={Indiana University Mathematics Journal},
  volume={21},
  number={5},
  pages={457--467},
  year={1971},
  publisher={JSTOR}
}

@article{araki1973golden,
  title   = {{Golden-Thompson and Peierls-Bogolubov Inequalities for a General von Neumann Algebra}},
  author  = {Araki, Huzihiro},
  journal = {Communications in Mathematical Physics},
  volume  = {34},
  number  = {3},
  pages   = {167--178},
  year    = {1973}
}

@article{thompson2014golden,
  title={{The Golden-Thompson inequality: Historical aspects and random matrix applications}},
  author={Thompson, Colin J},
  journal={Journal of Mathematical Physics},
  volume={55},
  number={2},
  year={2014}
}

@article{hiai1993golden,
  title={{The Golden-Thompson trace inequality is complemented}},
  author={Hiai, Fumio and Petz, D{\'e}nes},
  journal={Linear Algebra and its Applications},
  volume={181},
  pages={153--185},
  year={1993},
  publisher={Elsevier}
}

@article{ando1994log,
  title={{Log majorization and complementary Golden-Thompson type inequalities}},
  author={Ando, Tsuyoshi and Hiai, Fumio},
  journal={Linear Algebra and its Applications},
  volume={197},
  pages={113--131},
  year={1994},
  publisher={Elsevier}
}

@article{sutter2017multivariate,
  title={Multivariate trace inequalities},
  author={Sutter, David and Berta, Mario and Tomamichel, Marco},
  journal={Communications in Mathematical Physics},
  volume={352},
  number={1},
  pages={37--58},
  year={2017},
  publisher={Springer}
}

@inproceedings{randrianantoanina2024triple,
  title={{Triple operator version of the Golden-Thompson inequality for traces on von Neumann algebras}},
  author={Randrianantoanina, Narcisse},
  booktitle={Annales de l'Institut Fourier},
  volume={74},
  number={1},
  pages={193--233},
  year={2024}
}

@article{so1992equality,
  title={Equality cases in matrix exponential inequalities},
  author={So, Wasin},
  journal={SIAM Journal on Matrix Analysis and Applications},
  volume={13},
  number={4},
  pages={1154--1158},
  year={1992},
  publisher={SIAM}
}

@article{cohen1988spectral,
  title={Spectral inequalities for matrix exponentials},
  author={Cohen, Joel E},
  journal={Linear Algebra and its Applications},
  volume={111},
  pages={25--28},
  year={1988},
  publisher={Elsevier}
}

@article{cohen1982eigenvalue,
  title={Eigenvalue inequalities for products of matrix exponentials},
  author={Cohen, Joel E and Friedland, Shmuel and Kato, Tosio and Kelly, Frank P},
  journal={Linear Algebra and its Applications},
  volume={45},
  pages={55--95},
  year={1982},
  publisher={Elsevier}
}

@article{friedland1994product,
  title={On the product of matrix exponentials},
  author={Friedland, Shmuel and So, Wasin},
  journal={Linear Algebra and its Applications},
  volume={196},
  pages={193--205},
  year={1994},
  publisher={Elsevier}
}

@article{petz1994survey,
  title={A survey of certain trace inequalities},
  author={Petz, D{\'e}nes},
  journal={Banach Center Publications},
  volume={30},
  number={1},
  pages={287--298},
  year={1994}
}

@article{lieb1973convex,
  title={{Convex trace functions and the Wigner-Yanase-Dyson conjecture}},
  author={Lieb, Elliott H},
  journal={Les rencontres physiciens-math{\'e}maticiens de Strasbourg-RCP25},
  volume={19},
  pages={0--35},
  year={1973}
}

@article{johnson2002eigenvalues,
  title={Eigenvalues of words in two positive definite letters},
  author={Johnson, Charles R and Hillar, Christopher J},
  journal={SIAM Journal on Matrix Analysis and Applications},
  volume={23},
  number={4},
  pages={916--928},
  year={2002},
  publisher={SIAM}
}

@article{hillar2005positive,
  title={Positive eigenvalues and two-letter generalized words},
  author={Hillar, Christopher and Johnson, Charles R and Spitkovsky, Ilya M},
  journal={arXiv preprint math/0504573},
  year={2005}
}

\appendix

\section{Exact formula for the case $n=m=5$}

For completeness, we provide an explicit derivation of the formula for the case $n=m=5$.
Let
\begin{equation}
L(x):=\Tr(A_x^5B_x^5),
\qquad
R(x):=p_{5,5}(A_x,B_x),
\label{eq:LR}
\end{equation}
where, by definition,
\begin{equation}
\binom{10}{5}R(x)=[t^5s^5]\,\Tr(tA_x+sB_x)^{10}.
\label{eq:Rdef}
\end{equation}

\begin{proposition}
For every $x\ge 0$, the matrices $A_x$ and $B_x$ are positive semidefinite and
\begin{equation}
L(x)=32x^5+256x^{10},
\label{eq:Lformula}
\end{equation}
\begin{align}
R(x)=\frac{x^4}{126}\bigl(5&+1422x+1675x^2+3130x^3
+4875x^4+5930x^5+4881x^6\bigr).
\label{eq:Rformula}
\end{align}
Consequently,
\begin{align}
L(x)-R(x)=\frac{5x^4}{126}\bigl(5475x^6&-1186x^5-975x^4-626x^3
-335x^2+522x-1\bigr).
\label{eq:gapformula}
\end{align}
In particular, $L(10^{-3})<R(10^{-3})$, so the first inequality in Eq.~\eqref{eq:refinement} fails for $(A_{10^{-3}},B_{10^{-3}})$.
\end{proposition}

\begin{proof}
We first note that
\[
C:=\begin{pmatrix}1&-1\\-1&1\end{pmatrix}
\]
has eigenvalues $2$ and $0$, hence $C\succeq 0$. Since
\[
A_x=1\oplus xC,
\qquad
B_x=xC\oplus 1,
\]
we have $A_x,B_x\succeq 0$ for all $x\ge 0$.

Eq.~\eqref{eq:Rformula} is proved below. At $x=10^{-3}$, the factor in parentheses in Eq.~\eqref{eq:gapformula} is negative, so indeed $L(10^{-3})-R(10^{-3})<0$.
\end{proof}

\subsection{Derivation of Eq.~\eqref{eq:Rformula}}

Set
\begin{equation}
M:=tA_x+sB_x=
\begin{pmatrix}
t+sx&-sx&0\\
-sx&x(t+s)&-tx\\
0&-tx&tx+s
\end{pmatrix}.
\label{eq:Mmatrix}
\end{equation}
Then, by definition,
\[
\binom{10}{5}R(x)=[t^5s^5]\,\Tr(M^{10}).
\]
Let
\[
\chi_M(\lambda)=\det(\lambda I-M)=\lambda^3-e_1\lambda^2+e_2\lambda-e_3.
\]
The coefficients are the elementary symmetric polynomials in the eigenvalues of $M$.
From Eq.~\eqref{eq:Mmatrix} we compute
\[
e_1=\Tr(M)=(2x+1)(t+s).
\]
Next, $e_2$ is the sum of the principal $2\times 2$ minors. These are
\begin{align*}
\det\begin{pmatrix}t+sx&-sx\\-sx&x(t+s)\end{pmatrix} & = xt^2+x(x+1)ts,\\
\det\begin{pmatrix}t+sx&0\\0&tx+s\end{pmatrix} & = xt^2+(x^2+1)ts+xs^2,\\
\det\begin{pmatrix}x(t+s)&-tx\\-tx&tx+s\end{pmatrix} & = xs^2+x(x+1)ts.
\end{align*}
Summing them gives
\[
e_2=2x(t^2+s^2)+(3x^2+2x+1)ts.
\]
Finally,
\[
e_3=\det(M)=x(x+1)ts(t+s).
\]
It is convenient to introduce
\begin{align}
\alpha&:=2x+1,
&\beta&:=3x^2+2x+1,\nonumber\\
\delta&:=2x,
&\gamma&:=x(x+1),
\label{eq:alphabetadeltagamma}
\end{align}
and also
\begin{equation}
u:=t+s,
\qquad
v:=t^2+s^2,
\qquad
w:=ts.
\label{eq:uvw}
\end{equation}
Then
\begin{equation}
e_1=\alpha u,
\qquad
e_2=\delta v+\beta w,
\qquad
e_3=\gamma u w.
\label{eq:e123uvw}
\end{equation}
Let $a_n:=\Tr(M^n)$. The Newton identities for a $3\times 3$ matrix give the degree-$10$ power sum in terms of $e_1,e_2,e_3$:
\begin{align}
a_{10} = e_1^{10}&-10e_1^8e_2+10e_1^7e_3+35e_1^6e_2^2-60e_1^5e_2e_3\nonumber\\
&-50e_1^4e_2^3+25e_1^4e_3^2+100e_1^3e_2^2e_3+25e_1^2e_2^4\nonumber\\
&-60e_1^2e_2e_3^2-40e_1e_2^3e_3+10e_1e_3^3-2e_2^5+15e_2^2e_3^2.
\label{eq:newton10}
\end{align}
Thus the problem reduces to extracting the coefficient of $t^5s^5$ from the right-hand side of Eq.~\eqref{eq:newton10}.

\begin{observation}
\label{lem:coeffextract}
Let $a,b,c\ge 0$ be integers with $a+2b+2c=10$. Then
\[
[t^5s^5]u^av^bw^c
=\sum_{j=0}^b \binom{b}{j}\binom{a}{5-c-2j},
\]
where $\binom{a}{m}$ is understood to be $0$ when $m\notin\{0,1,\dots,a\}$.
\end{observation}

\begin{proof}
From Eq.~\eqref{eq:uvw},
\begin{align*}
v^b & =(t^2+s^2)^b=\sum_{j=0}^b \binom{b}{j} t^{2j}s^{2(b-j)},\\
u^a & =(t+s)^a=\sum_{r=0}^a \binom{a}{r} t^r s^{a-r}.
\end{align*}
Therefore
\[
u^av^bw^c
=\sum_{j=0}^b\sum_{r=0}^a \binom{b}{j}\binom{a}{r}
\, t^{r+2j+c}s^{a-r+2(b-j)+c}.
\]
The coefficient of $t^5s^5$ is obtained when
\[
r+2j+c=5,
\]
that is, when $r=5-c-2j$. Substituting this value gives the stated formula.
\end{proof}

Applying Observation~\ref{lem:coeffextract} to the monomials that occur in Eq.~\eqref{eq:newton10}, we obtain the following identities:
\begin{align}
[t^5s^5]u^{10}&=252,
\label{eq:coeff1}\\
[t^5s^5]u^8(\delta v+\beta w)&=112\delta+70\beta,
\label{eq:coeff2}\\
[t^5s^5]u^8\gamma w&=70\gamma,
\label{eq:coeff3}\\
[t^5s^5]u^6(\delta v+\beta w)^2&=52\delta^2+60\beta\delta+20\beta^2,
\label{eq:coeff4}\\
[t^5s^5]u^6(\delta v+\beta w)\gamma w&=30\delta\gamma+20\beta\gamma,
\label{eq:coeff5}\\
[t^5s^5]u^4(\delta v+\beta w)^3&=24\delta^3+42\beta\delta^2+24\beta^2\delta+6\beta^3,
\label{eq:coeff6}\\
[t^5s^5]u^6\gamma^2w^2&=20\gamma^2.
\label{eq:coeff7}
\end{align}
A second application gives
\begin{align}
[t^5s^5]u^4(\delta v+\beta w)^2\gamma w&=14\delta^2\gamma+16\beta\delta\gamma+6\beta^2\gamma,
\label{eq:coeff8}\\
[t^5s^5]u^2(\delta v+\beta w)^4&=12\delta^4+24\beta\delta^3+24\beta^2\delta^2+8\beta^3\delta+2\beta^4,
\label{eq:coeff9}\\
[t^5s^5]u^4(\delta v+\beta w)\gamma^2w^2&=8\delta\gamma^2+6\beta\gamma^2,
\label{eq:coeff10}\\
[t^5s^5]u^2(\delta v+\beta w)^3\gamma w&=6\delta^3\gamma+12\beta\delta^2\gamma+6\beta^2\delta\gamma+2\beta^3\gamma,
\label{eq:coeff11}\\
[t^5s^5]u^4\gamma^3w^3&=6\gamma^3,
\label{eq:coeff12}\\
[t^5s^5](\delta v+\beta w)^5&=30\beta\delta^4+20\beta^3\delta^2+\beta^5,
\label{eq:coeff13}\\
[t^5s^5]u^2(\delta v+\beta w)^2\gamma^2w^2&=4\delta^2\gamma^2+4\beta\delta\gamma^2+2\beta^2\gamma^2.
\label{eq:coeff14}
\end{align}

\clearpage
\noindent Substituting Eqs.~\eqref{eq:e123uvw} and \eqref{eq:coeff1}--\eqref{eq:coeff14} into Eq.~\eqref{eq:newton10} yields
\begin{align}
[t^5s^5]a_{10}&=252\alpha^{10}-10\alpha^8(112\delta+70\beta)\nonumber\\
&+10\alpha^7(70\gamma)
+35\alpha^6(52\delta^2+60\beta\delta+20\beta^2) \notag\\
& -60\alpha^5(30\delta\gamma+20\beta\gamma)\nonumber\\
&-50\alpha^4(24\delta^3+42\beta\delta^2+24\beta^2\delta+6\beta^3)\nonumber\\
&+25\alpha^4(20\gamma^2)+100\alpha^3(14\delta^2\gamma+16\beta\delta\gamma+6\beta^2\gamma)\nonumber\\
&+25\alpha^2(12\delta^4+24\beta\delta^3+24\beta^2\delta^2+8\beta^3\delta+2\beta^4) \notag\\
& -60\alpha^2(8\delta\gamma^2+6\beta\gamma^2)\nonumber\\
&-40\alpha(6\delta^3\gamma+12\beta\delta^2\gamma+6\beta^2\delta\gamma+2\beta^3\gamma) \notag\\
& +10\alpha(6\gamma^3)-2(30\beta\delta^4+20\beta^3\delta^2+\beta^5)\nonumber\\
&+15(4\delta^2\gamma^2+4\beta\delta\gamma^2+2\beta^2\gamma^2).
\label{eq:longcoefficient}
\end{align}
Now substitute the values of $\alpha,\beta,\delta,\gamma$ from Eq.~\eqref{eq:alphabetadeltagamma}. A direct simplification gives
\begin{align}
[t^5s^5]a_{10}=2x^4\bigl(4881x^6&+5930x^5+4875x^4+3130x^3+1675x^2+1422x+5\bigr).
\label{eq:coefffinal}
\end{align}
Since $\binom{10}{5}=252$, Eqs.~\eqref{eq:Rdef} and \eqref{eq:coefffinal} imply
\begin{align*}
R(x)&=\frac{1}{252}[t^5s^5]a_{10}\\
&=\frac{x^4}{126}\bigl(5+1422x+1675x^2+3130x^3+4875x^4+5930x^5+4881x^6\bigr).
\end{align*}

\end{document}